\documentclass[12pt,halfline,a4paper]{article}
\usepackage{enumerate}
\usepackage{amsfonts}
\usepackage{amsthm,amssymb,amsmath}
\usepackage{mathrsfs,MnSymbol,setspace, xcolor}
\usepackage[colorlinks=true, linkcolor = blue, citecolor = cyan]{hyperref}
\usepackage[top=3cm,right=3cm,bottom=3cm,left=3cm]{geometry}                 
\usepackage{amsfonts}
\usepackage{mathrsfs}
\usepackage{verbatim}
\usepackage{graphicx}
\usepackage{tikz}
\usetikzlibrary{arrows}

\usepackage{soul}

\author{Doratossadat Dastgheib\footnote{d\_dastgheib@sbu.ac.ir} , Hadi Farahani\footnote{h$_-$farahani@sbu.ac.ir}}

\newtheorem{proposition}{\textbf{Proposition}}[section]
\newtheorem{lemma}{\textbf{Lemma}}[section]
\newtheorem{corollary}{\textbf{Corollary}}[section]

\newtheorem{theorem}{\textbf{Theorem}}[section]
\newenvironment{definition}{\textbf{Definition.}}{\hfill$\blacktriangleleft$}

\newcounter{nremark}
\newtheorem{remark}{\refstepcounter{nremark}\textbf{Remark}}{}

\newcommand{\LB}[1]{(\L$_B$#1)}

\numberwithin{equation}{section}

\author{Doratossadat Dastgheib\footnote{d\_dastgheib@sbu.ac.ir} , Hadi Farahani\footnote{h$_-$farahani@sbu.ac.ir}}
\title{Some Doxastic \L{ukasiewicz}  Logics	
	}

\begin{document}
\maketitle
\begin{abstract}

We propose a  doxastic \L ukasiewicz logic \textbf{B\L} that is sound and complete with respect to the class of   Kripke-based models in which  atomic propositions and  accessibility relations are both infinitely valued in the standard  MV-algebra [0,1]. We also introduce some extensions of \textbf{B\L} corresponding to   axioms \textbf{D}, \textbf{4}, and \textbf{T} of classical epistemic logic. Furthermore,  completeness of these extensions  are established corresponding to the appropriate classes of  models.

\end{abstract}
 
{\small \textit{keywords:} \L ukasiewicz Logic, Doxastic, Kripke Models, Completeness}

\section{Introduction}
%

 Interpreting modal operators that have  
 notions of \textit{belief} and \textit{knowledge} is challenging. One of the pioneers in formally defining these concepts is Hintikka, as demonstrated in his book \cite{Hintikka1962}. Subsequently, Fagin et al. further ventured into the realm of knowledge-based reasoning in 1995, delving into its practical applications across various domains \cite{Fagin1995}.
Previously, significant attention has been directed towards the exploration of many-valued modal logics. Various researchers have introduced finite many-valued modal logics such as \cite{Ftting1991,Fitting1992, Morgan1979, Morikawa1989, Ostermann1988, Segerberg1967, Thomason1978}.

Recent years have witnessed  endeavors to enrich many-valued logics with modal operators. In \cite{Metcalfe2009}, a basic propositional G\"odel modal logic is addressed.
In \cite{Caicedo2015,Caicedo2010}, Caicedo and et. al.  proposed some bi-modal G\"odel and G\"odel modal logics and proved their  completeness based on Kripke-style semantics where both propositions and accessibility relations are infinite-valued. They also provided semantics for finite model property to show the decidability of desired logics \cite{Caicedo2013}.
Similarly to the previously mentioned research, \cite{Dastgheib2020} proposed several  extensions of epistemic G\"odel logic, incorporating both fuzzy atomic propositions and fuzzy accessibility relations. Prior to this work,  \cite{Gounder1998} also tackled the topics of fuzzy epistemic and fuzzy deontic logics.
  A bi-modal logic axiomatization  arising from  G\"odel-Kripke structures with accessibility relations that take crisp values was established in \cite{Rodriguez2020}. 
 
 Some many-valued generalizations of modal logic which their Kripke-based semantics are structured with  crisp relations in which variables take their truth values in the MV-algebra [0, 1] introduced in \cite{Hansoul2011}. Some modal extensions of product fuzzy logic are investigated in \cite{Vidal2015}. 
 Di Nola and et. al. in \cite{Nola2015}, established three-valued multi-modal \L ukasiewicz logic and proved its completeness.
Recently, in \cite{Vidal2021} it is shown that a large family of modal \L ukasiewicz logics are not axiomatizable.

In this paper, we present a  doxastic \L ukasiewicz logic \textbf{B\L} and its some extensions corresponding to  axioms \textbf{D}, \textbf{4},  and \textbf{T} in classical epistemic logic. We add some modal (belief) operator $ B$ to  the language of  \L ukasiewicz  logic. Our proposed semantics is Kripke-based  such that the  atomic propositions and  accessibility relations both take values in the standard  MV-algebra [0,1]. The interpretation of belief operator $B$ is inspired by the truth definition of belief in classical epistemic logic, with some fuzzy adaptations as seen in \cite{Gounder1998} and \cite{Dastgheib2020}.

This paper is organized as follows. In section \ref{preliminiaries}, we provide a review of the axioms of \L ukasiewicz logic that will be utilized throughout this paper. In section \ref{general_belief}, we first  introduce the language of doxastic \L ukasiewicz logic, which expands the \L ukasiewicz logic language by incorporating a modal (belief) operator.
 Then present a  Kripke-based semantics in which both atomic propositions and accessibility relations take fuzzy values. Finally, we propose several axiomatic systems and demonstrate their  soundness and completeness relative to the corresponding semantics. 

\section{Preliminaries} \label{preliminiaries}

In this section, we give a short overview of propositional basic fuzzy logic (\textbf{BL}) and \L ukasiewicz logic from \cite{Hajek1998}. 
In the following, axioms (A1) to (A7) are referred to as the \textbf{BL} axioms, while (A8) represents a scheme that can be established using the axioms of \textbf{BL}. 
\begin{equation*}
\begin{array}{ll}
(A1) (\varphi \rightarrow \psi) \rightarrow ((\psi \rightarrow \chi) \rightarrow (\varphi \rightarrow \chi)) & (A5) (\varphi \rightarrow (\psi \rightarrow \chi)) \leftrightarrow ((\varphi \, \& \, \psi) \rightarrow \chi) \\
(A2) (\varphi \, \& \, \psi)\rightarrow \varphi & (A6) ((\varphi \rightarrow \psi)\rightarrow \chi) \rightarrow (((\psi \rightarrow \varphi)\rightarrow \chi)\rightarrow \chi)\\
(A3) (\varphi \, \& \, \psi) \rightarrow (\psi \, \& \, \varphi) & (A7) \perp \rightarrow \varphi \\ 
(A4) (\varphi \, \& \, (\varphi \rightarrow \psi)) \rightarrow (\psi \, \& \, (\psi \rightarrow \varphi)) & (A8) \varphi \rightarrow (\psi \rightarrow (\varphi \,\&\, \psi))\\
\end{array}
\end{equation*}

Propositional \L ukasiewicz logic; denoted by \textbf{\L}; is an extension of \textbf{BL}  equipped with the double negation axiom $\neg \neg \varphi \rightarrow \varphi$.
It is possible to demonstrate that all axioms of \textbf{\L} can be derived from the following four axioms:
\begin{itemize}
\item[(\L{}1)] $\varphi \rightarrow (\psi \rightarrow \varphi)$
\item[(\L{}2)] $(\varphi \rightarrow \psi) \rightarrow ((\psi \rightarrow \chi) \rightarrow (\varphi \rightarrow \chi))$
\item[(\L{}3)] $(\neg \varphi \rightarrow \neg \psi) \rightarrow (\psi \rightarrow \varphi)$
\item[(\L{}4)] $((\varphi \rightarrow \psi) \rightarrow \psi) \rightarrow ((\psi \rightarrow \varphi) \rightarrow \varphi)$
\end{itemize}

If we consider connective $\leftrightarrow$ as defined in
\cite{Hajek1998}, then  $\neg \neg \varphi \leftrightarrow \varphi$ is a theorem of \textbf{\L} which results that $\neg \neg \varphi$ and $\varphi$ are provably equivalent.
In the following section, we outline several properties of propositional \L ukasiewicz logic, which will be utilized throughout this paper. Specifically, we refer to the properties identified as (\L12), (\L13), and (\L14) as a form of \textit{replacement} within the context of this paper. For additional details please see \cite{Hajek1998}.
$$\begin{array}{llcll}
	(\text{\L}5) & \neg (\varphi \,\&\, \psi) \leftrightarrow (\neg \varphi \veebar \neg \psi) & \quad &
	& \\
	(\text{\L}6) & \neg (\varphi \veebar \psi) \leftrightarrow (\neg \varphi \, \& \, \neg \psi) & \quad &
	&  \\
	(\text{\L{}}7)& (\varphi \veebar \psi) \leftrightarrow (\neg \varphi \rightarrow \psi) & \quad&
	&\\
	(\text{\L}8)& \varphi \rightarrow (\varphi \veebar \psi) & \quad&
	& \\
	(\text{\L{}}9)& (\varphi \,\&\, (\varphi \rightarrow \psi)) \rightarrow \psi & \quad & & \\
	(\text{\L{}}10) & 	\multicolumn{4}{l}{((\varphi_1 \rightarrow \psi_1) \,\&\, (\varphi_2 \rightarrow \psi_2))\rightarrow((\varphi_1 \,\&\, \varphi_2)\rightarrow (\psi_1\,\&\, \psi_2))}\\
	
	(\text{\L{}}11) & \neg \neg \varphi \leftrightarrow \varphi & \quad &
	& \\
	
	(\text{\L{}}12) & (\varphi \leftrightarrow  \psi) \rightarrow ((\varphi \rightarrow \chi) \leftrightarrow (\psi \rightarrow \chi)) & \quad &
	& \\
	
	(\text{\L{}}13) &(\varphi \leftrightarrow  \psi) \rightarrow ((\chi \rightarrow \varphi) \leftrightarrow  (\chi \rightarrow \psi))& \quad &
	& \\
	
	(\text{\L{}}14) &(\varphi \leftrightarrow  \psi) \rightarrow ((\varphi \,\&\, \chi) \leftrightarrow  (\psi \,\&\, \chi))& \quad &
	&  \\
	
	(\L15) &(\varphi \rightarrow \perp) \leftrightarrow \neg \varphi & & & 
\end{array}$$
\vskip 0.1cm

In this paper, we consistently represent the set of atomic propositions as $\mathcal{P}$ and the set of agents as $\mathcal{A}$. Additionally, we utilize the symbol $\perp$ as a propositional constant, which always has the value of $0$.
\\

\section{Some Doxastic Extensions of \L{ukasiewicz } Logic } \label{general_belief}

In this section, first we expand the language of propositional \L ukasiewicz logic with a belief operator. Next, we introduce the concept of a D\L L-model. 
Also, we propose some axiomatic systems  that expand upon \L ukasiewicz logic and  demonstrate both the soundness and completeness of these logical systems.
\subsection{Syntax and Semantics}
\begin{definition}
The \textit{ doxastic \L ukasiewicz language}
; in short D\L L, is defined using the following BNF:
$$\varphi::= p \;|\; \neg \varphi\;|\; \varphi \, \& \, \varphi \;|\; \varphi \rightarrow \varphi \;|\; B_a \varphi  $$
where $p\in\mathcal{P}$ and $a\in \mathcal{A}$. The other connectives $\veebar$, $\wedge$ and $\vee$ can be defined similar to the \L ukasiewicz logic as follows:
$$\begin{array}{ll}
\varphi \veebar \psi = \neg \varphi \rightarrow \psi, &
\varphi \wedge \psi = \varphi \,\&\, (\varphi \rightarrow \psi),\\
\multicolumn{2}{l}{\varphi \vee \psi = ((\varphi \rightarrow \psi) \rightarrow \psi ) \wedge ((\psi \rightarrow \varphi)\rightarrow \varphi).}
\end{array}$$
\end{definition}

For simplicity we use sub-index ``a" in the provided examples. Henceforth, we will use $B \varphi$ and $r(s,s')$ instead of $B_a \varphi$ and $r_a(s,s')$, respectively.
\\

	\begin{definition} 
		A \textit{doxastic \L ukasiewicz logic model} or in short \textit{D\L L-model} is a tuple $\mathfrak{M}= (S, r_{a_{|a\in \mathcal{A}}}, \pi)$ in which $S$ is a set that includes the states of  the model, $r_{a_{|a\in \mathcal{A}}}: S\times S \rightarrow [0,1]$ is the accessibility relation, and $\pi: S\times \mathcal{P} \rightarrow [0,1]$ is the valuation function that assigns a value to each proposition in each state.
	\end{definition}
	\vskip 0.1cm
	In contrast to the classical Kripke structures, the above definition incorporates fuzzy values for both accessibility relations and atomic propositions. The values assigned to the accessibility relations reflect the ability of agents to differentiate between two arbitrary states within the model. For any given agent $a\in \mathcal{A}$ and states $s$ and $s'\in S$, a higher value of $r_a(s,s')$ implies that agent $a$ has greater difficulty distinguishing between states $s$ and $s'$. Consequently, when $r_a(s,s')=0$, it signifies that the agent can fully discern the distinction between the two states, while $r_a(s,s')=1$ indicates that states $s$ and $s'$ are entirely indistinguishable for agent $a$.

\begin{definition}
Let $\mathfrak{M} = (S, r, \pi)$ be a D\L L-model. 
For each  D\L L-formula $\varphi$ and a state $s\in S$, $V(s,\varphi)$ denotes  the extension of the valuation function, defined recursively as follows: (for simplicity we use $V_s(\varphi)$ instead of $V(s,\varphi)$)
\begin{align*}
& V_s(p) = \pi(s,p) \;\; \forall p\in \mathcal{P}, & \\
& V_s(\neg \varphi) = 1- V_s(\varphi), & \\
& V_s(\varphi \,\&\, \psi) = \max\{0, V_s(\varphi)+V_s(\psi) - 1\}, & \\
& V_s(\varphi \rightarrow \psi) = \min \{1, 1- V_s(\varphi)+V_s(\psi)\},& \\
& V_s(B  \varphi) = \inf_{s'\in S} \max\{1-r (s,s'),V_{s'}(\varphi)\}. &
\end{align*}
Values of $\veebar$, $\wedge$, and $\vee$ are defined as the same as \L ukasiewicz logic as follows:
\begin{align}
\label{Luka_veebar}
& V_s(\varphi \veebar \psi) = \min\{1,V_s(\varphi)+V_s(\psi)\}, & \\ 
\label{Luka_wedge}
& V_s(\varphi \wedge \psi) = \min \{V_s(\varphi) , V_s(\psi)\}, & \\ 
\label{Luka_vee}
& V_s(\varphi \vee \psi) = \max\{V_s(\varphi) , V_s(\psi)\}. & 
\end{align}
\end{definition}

\begin{remark}
The definition of the truth value of belief in classical epistemic logic is as follows:
\begin{equation}\label{eq_imp2}
V_s(B  \varphi)=1 \;\iff\; \forall s'\in S\;\; (r (s,s')=1 \Rightarrow V_{s'}(\varphi)=1).
\end{equation}
As we see in (\ref{eq_imp2}), the classical definition of $B\varphi$ assures that an agent believes in $\varphi$ in a state $s$ if and only if $\varphi$ is valid in all states $s'$ that are accessible from $s$.
In classical logic, the right side of (\ref{eq_imp2}) is equivalent to $\neg (r(s,s')=1) \vee (V_{s'}(\varphi)=1)$ for all $s'\in S$. Following this idea we defined  the semantics of belief as the infimum of all $\neg r(s,s') \vee V_{s'} (\varphi)$ for all $s'\in S$ in the definition above. Also,
%
note that in this definition 
 when an agent $a$ can distinguish between two states $s$ and $s'$, her/his belief in state $s$ is entirely independent of the value of $\varphi$ in state $s'$.
\end{remark}

	In the following, we give some properties of fuzzy relations in Kripke models.
	\\
	\begin{definition}
		The D\L L-model $\mathfrak{M} = (S, r_{a_{|a\in \mathcal{A}}}, \pi)$ is called:
		\begin{itemize}
			\item \textit{serial} if for all $a\in \mathcal{A}$ and all $s\in S$, there exists a state $s'\in S$ such that $r_a(s,s')=1$.
			\item \textit{reflexive} if for all $a\in \mathcal{A}$ and all $s\in S$, we have $r_a(s,s)=1$.
			\item \textit{transitive} if for all $a\in \mathcal{A}$ and all $s,s',s''\in S$, $r_a(s,s'') \geq \min \{r_a(s,s'), r_a(s',s'')\}$.
			\item \textit{recognizable} if for all $a\in \mathcal{A}$ and all $s,s'\in S$, we have \mbox{$r_a(s,s')\leq 0.5$}.
		\end{itemize}
		If $\mathcal{M}$ represents a class of D\L L-models, we say $\mathcal{M}$ is serial if, for all $\mathfrak{M}\in \mathcal{M}$, $\mathfrak{M}$ is serial. Similarly, we call $\mathcal{M}$ reflexive, transitive, or recognizable  if, for all $\mathfrak{M} \in \mathcal{M}$, $\mathfrak{M}$ is reflexive, transitive, or recognizable, respectively.
\end{definition}

\subsection{Soundess and Completeness} \label{gb_sound_complete}

\begin{definition} \label{def_validity}
Let $\varphi$ be a D\L L-formula and $\mathfrak{M} = (S, r_{a_{|a\in \mathcal{A}}}, \pi)$ be a D\L L-model. We say that $\varphi$ is \textit{valid in a pointed model} $(\mathfrak{M},s)$ if $V_s(\varphi)=1$ and denote it by $(\mathfrak{M},s)\vDash \varphi$. If for all $s\in S$ we have $(\mathfrak{M},s)\vDash \varphi$, then we call it $\mathfrak{M}$\textit{-valid} and use the notation $\mathfrak{M} \vDash \varphi$. If $\varphi$ is $\mathfrak{M}$-valid for all models $\mathfrak{M}$ in a class of models $\mathcal{M}$, we say that $\varphi$ is $\mathcal{M}$-\textit{valid}, and show it by $\mathcal{M}\vDash \varphi$.
We use the  notation $\vDash \varphi$, if for all models $\mathfrak{M}$ we have $\mathfrak{M} \vDash \varphi$ and call the formula $\varphi$  a \textit{valid} formula.

\end{definition}

\begin{proposition} \label{distribution_ax}
For all D\L L-formulae $\varphi$ and $\psi$,
the following schemes are valid.
\begin{enumerate}
\item $(B  \varphi \,\&\, B  (\varphi \rightarrow \psi) )\rightarrow B  \psi  $
\item $B  (\varphi \rightarrow \psi) \rightarrow (B  \varphi \rightarrow B  \psi)$
\end{enumerate}
\end{proposition}
\begin{proof}
Let $\mathfrak{M} = (S, r_{a_{|a\in \mathcal{A}}},\pi)$ be an arbitrary D\L L-model.
In order to prove part (1), first note that
we have the following statement for all $s'\in S$
$$\max\{1-r(s,s'), V_{s'}(\varphi)\}+ \max\{1-r(s,s'), min\{1,1-V_{s'}(\varphi), V_{s'}(\psi)\}\}-1 \leq \max\{1-r(s,s'), V_{s'}(\psi)\}.$$
Thus, we have
\begin{eqnarray*}
&\inf_{s' \in S}\max\{1-r(s,s'), V_{s'}(\varphi)\}+ \inf_{s' \in S}\max\{1-r(s,s'), min\{1,1-V_{s'}(\varphi)+ V_{s'}(\psi)\}\} -1 & \\
&\leq \inf_{s' \in S}\max\{1-r(s,s'), V_{s'}(\psi)\},&
\end{eqnarray*}
 and, therefore
\begin{eqnarray*}
&\max[0, \inf_{s' \in S}\max\{1-r(s,s'), V_{s'}(\varphi)\}+ \inf_{s' \in S}\max\{1-r(s,s'), min\{1,1-V_{s'}(\varphi)+ V_{s'}(\psi)\}\}-1]& \\
&\leq \inf_{s' \in S}\max\{1-r(s,s'), V_{s'}(\psi)\}.&
\end{eqnarray*}	
The left side is $V_s(B\varphi \,\&\,B (\varphi \rightarrow \psi))$, and the right side is $V_s(B \psi)$. So,
$V_s(B\varphi \,\&\,B (\varphi \rightarrow \psi)) \leq V_s(B \varphi)$ and hence
$(B \varphi \,\&\, B(\varphi \rightarrow \psi)) \rightarrow B\psi$ is valid.

\vskip 0.1cm 
The validity of  part (2) comes from the equivalency of this scheme and part (1) using axiom (A5)  of BL.
\end{proof}
\begin{proposition} \label{notvalid_gb}
The following schemes are not valid.
\begin{enumerate}
\item $\neg B  \perp$
\item $B  \varphi \rightarrow B  B  \varphi$
\item $\neg B  \varphi \rightarrow B  \neg B  \varphi$
\item $B  \varphi \rightarrow \varphi$
\end{enumerate}
\end{proposition}
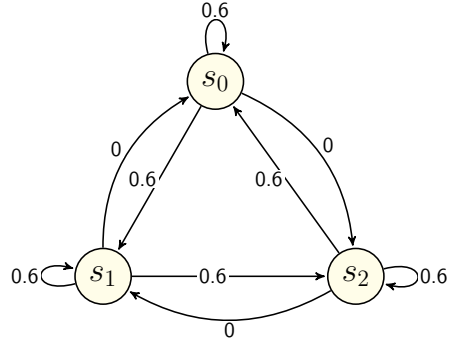
\begin{figure} 
\begin{center}
\resizebox*{6cm}{!}{
	
	\begin{tikzpicture}[->,>=stealth',shorten >=1pt,auto,
		thick, main node/.style={circle,fill=yellow!10,draw,
			font=\sffamily\Large\bfseries,minimum size=10mm}]
		\node[main node] (s0) at (1.5,2) {$s_0$};
		\node[main node] (s1) at (-0.5,-1.5) {$s_1$};
		\node[main node] (s2) at (4,-1.5) {$s_2$};
		
		\path[every node/.style={font=\sffamily\small,fill=white,inner sep=1pt}]
		(s0) edge [loop above] node {0.6} (s0)
		(s1) edge [loop left] node {0.6} (s1) 
		(s2) edge [loop right] node {0.6} (S2)
		(s0) edge [left] node {0.6} (s1)
		(s1) edge [bend left] node {0} (s0)
		(s0) edge [bend left] node {0} (s2)
		(s2) edge [left] node {0.6} (s0)
		(s1) edge [left] node {0.6} (s2)
		(s2) edge [bend left] node {0} (s1)
		;
	\end{tikzpicture}
	
}
\end{center}
\caption{The diagram of $\mathfrak{M}_1$ model.} \label{fig001}
\end{figure}
\begin{proof}
Let $\mathcal{A}=\{a\}$ and $ \mathcal{P}=\{p\}$  be the set of agents and atomic propositions respectively, and assume $\mathfrak{M}_1 = (S, r  , \pi )$ be a D\L L-model (See Figure \ref{fig001}), where $S= \{s_0,s_1,s_2\}$ such that for all $s_i,s_j\in S$ if $i\leq j$, then $r (s_i,s_j)=0.6$, otherwise $r (s_i,s_j)=0$. Let the valuation function $\pi$ is defined as $\pi(s_0, p) =0.8$, $\pi(s_1,p)=0.9$ and $\pi(s_2,p)=0.7$.
First, we obtain the value of $\neg B  \perp$ at state $s_0$:
$$\begin{array}{ll}
V_{s_0}(\neg B  \perp) & = 1- V_{s_0}(B  \perp),\\
& = 1-  \inf \{ \max \{0.4, 0\}, \max\{1,0\}, \max\{0.4 ,0\}\}, \\
& = 0.6.
\end{array} $$
For parts (2) and (3), after substituting $\varphi$ with $p$ at state $s_0$, we obtain:
$$\begin{array}{ll}
V_{s_0}(B  p \rightarrow B  B  p) & = \min\{1, 1-V_{s_0}(B  p)+V_{s_0}(B  B  p)\}, \\
& = \min \{ 1, 1- 0.8 + 0.7\}, \\
& = 0.9,
\end{array}$$

$$\begin{array}{ll}
V_{s_0}(\neg B  p \rightarrow B  \neg B  p) & = \min \{1, 1 - 1 + V_{s_0}(B  p) + V_{s_0}(B  \neg B  p)\},\\
& = \min \{ 1, 0.2 + 0.4 \}, \\
& = 0.6.
\end{array}$$
For part (4), let $\pi(s_i,q)=0.2$ for all $s_i\in S$. It's not hard to check that $V_{s_0}(B  q\rightarrow q) = 0.8$.
\end{proof}

\begin{lemma} \label{gb_structures_properties}
Let $\varphi$ be a D\L L-formula and $\mathfrak{M} = (S, r, \pi)$ be a D\L L-model.
\begin{enumerate}
\item  If $\mathfrak{M}$ is serial, then $\mathfrak{M} \vDash \neg B  \perp,$
\item If $\mathfrak{M}$ is reflexive, then $\mathfrak{M} \vDash B  \varphi \rightarrow \varphi,$
\item If $\mathfrak{M}$ is transitive, then $\mathfrak{M} \vDash B  \varphi \rightarrow B  B  \varphi$,
\item If $\mathfrak{M}$ is recognizable, then $\mathfrak{M} \vDash \neg B  \varphi \rightarrow B  \neg B  \varphi$.
\end{enumerate}
\end{lemma}

\begin{proof}
Let $s\in S$ be an arbitrary state.

(1): By definition we have 
$V_s(\neg B  \perp) = 1-\inf_{s'\in S} \max \{1-r (s,s'), V_{s'}(\perp)\}.$
If $\mathfrak{M}$ is serial, then there is a state $s_0\in S$ such that $r (s,s_0)=1$. So $\max\{1-r (s,s_0),V_{s'}(\perp)\}$ takes value zero, then $V_s(\neg B  \perp) = 1$.

(2):
By definition  
$V_s(B  \varphi) = \inf_{s' \in S}\max\{1-r (s,s'), V_{s'}(\varphi)\}$,
and we have
$$\max\{1-r (s,s), V_{s}(\varphi)\} \geq \inf_{s'\in S}\max\{1-r (s,s'), V_{s'}(\varphi)\}.$$
From reflexivity we have $r (s,s) = 1$, so 
$$\max\{0, V_{s}(\varphi)\} \geq \inf_{s'\in S}\max\{1-r (s,s'), V_{s'}(\varphi)\},$$
which means that
$V_{s}(\varphi) \geq V_s(B  \varphi)$.

(3):
If $\mathfrak{M}$ is transitive then by definition for all $a\in \mathcal{A}$ and all $s,s',s'' \in S$ we have $r (s,s'')\geq \min\{r (s,s'), r (s',s'')\}$. It follows that
\begin{equation}
1-r (s,s'') \leq \max \{1-r (s,s'), 1-r (s',s'')\}.
\end{equation}	 
Thus 
$$\max\{1-r (s,s''), V_{s''}(\varphi)\}\leq \max \{1-r (s,s'), \max\{1-r (s',s''), V_{s''}(\varphi)\}\}.$$
So, by applying infimum on both sides we have:
\begin{eqnarray*}
& \inf_{s''\in S}\max\{1-r (s,s''), V_{s''}(\varphi)\}&\leq \inf_{s''\in S}\max \{1-r (s,s'), \max\{1-r (s',s''), V_{s''}(\varphi)\}\}\\
& & =  \max\{1-r (s,s'), \inf_{s''\in S}\max\{1-r (s',s''), V_{s''}(\varphi)\}\}.
\end{eqnarray*}
So, 
\begin{equation} \label{eq008}
\inf_{s''\in S}\max\{1-r (s,s''), V_{s''}(\varphi)\} \leq \inf_{s' \in S} \max\{1-r (s,s'), \inf_{s''\in S}\max\{1-r (s',s''), V_{s''}(\varphi)\}\}.
\end{equation}

It is not hard to check that the left and right sides of \ref{eq008} are equal to $V_s(B  \varphi)$ and $V_s(B  B  \varphi)$, respectively, and thus $V_s(B  \varphi \rightarrow B  B  \varphi)=1.$ 

(4): Let $\mathfrak{M}$ be a recognizable model and $s\in S$ be an arbitrary state. We have
$$\begin{array}{ll}
V_s(\neg B  \varphi \rightarrow B  \neg B  \varphi) & = \min \{1, 1- V_s(\neg B  \varphi ) + V_s(B  \neg B  \varphi) \}\\
& = \min \{1,1-1+V_s(B  \varphi) + V_s(B  \neg B  \varphi)\}\\
& = \min\{1, V_s(B  \varphi)+V_s(B  \neg B  \varphi)\}.
\end{array}$$
So we just need to check that $V_s(B  \varphi)+V_s(B  \neg B  \varphi)\geq 1$. 
It is easy to check that $V_s(B  \varphi) \geq 0.5$ and $V_s(B  \neg B \varphi) \geq 0.5$. Thus $V_s(B  \varphi)+V_s(B  \neg B  \varphi)\geq 1$ and $V_s(\neg B  \varphi \rightarrow B  \neg B  \varphi)=1$.
\end{proof}

In the following we propose some axiom schemes and inference rules to introduce some axiomatic systems. Suppose that $\varphi$ and $\psi$ be D\L L-formulae.

\begin{itemize}
\item[(\L$_B$0)] All instances of tautologies in propositional \L ukasiewicz logic
\item[(\L$_B$1)] $(B  \varphi \, \& \, B (\varphi \rightarrow \psi)) \rightarrow (B  \psi)$
\item[(\L$_B$2)] $\neg B  \perp$
\item[(\L$_B$3)] $B  \varphi \rightarrow B  B  \varphi$
\item[(\L$_B$4)] $\neg B  \varphi \rightarrow B  \neg B  \varphi$
\item[(\L$_B$5)] $B  \varphi \rightarrow \varphi$
\item[(R$_{\text{MP}}$)] $\cfrac{\varphi \qquad \varphi \rightarrow \psi}{\psi}$
\item[(R$_\text{GB}$)] $\cfrac{\varphi}{B  \varphi}$ 
\end{itemize}

\LB{1} serves as the distribution axiom for belief. 
Axiom (\L$_B$2), denoted as \textbf{ D} in classical epistemic logic, ensures that agents act rationally, meaning each agent does not believe in any contradictions. 
The positive and negative introspection axioms, typically found in classical epistemic logic as axioms \textbf{4}, \textbf{5} (see \cite{Ditmarsch2008}), correspond to (\L$_B$3) and (\L$_B$4), respectively. The axiom (\L$_B$5) is similar to the axiom of truth \textbf{T} in classical epistemic logic.

Let \textbf{B\L} be our basic axiomatic system that is an extension of propositional \L ukasiewicz logic containing axioms \LB{0}, \LB{1} and inference rules (R$_{\text{MP}}$) and (R$_\text{GB}$). 
We can extend the system \textbf{B\L} by incorporating certain axioms from the aforementioned list. To refer to these extensions, we can use subscripts \textbf{D}, \textbf{4}, \textbf{5}, and \textbf{T}, which correspond to \LB{2}, \LB{3}, \LB{4}, and \LB{5}, respectively. For instance, $\textbf{B\L}_{\textbf{D45}}$ represents an axiomatic system obtained from \textbf{B\L} by adding \LB{2}, \LB{3}, and \LB{4}.


\begin{lemma} \label{gb_admissible_rules}
The inference rules (R$_{\text{MP}}$) and (R$_\text{GB}$) are sound, that is  if the premises of (R$_{{\text{MP}}}$) or (R$_{\text{GB}}$) are considered valid, then their conclusions can also be deemed valid.
\end{lemma}

\begin{proof}
It's obvious by the definition.
\end{proof}

\begin{definition}	\label{def_axiomatic_systems}
Let \textbf{A} shows the system \textbf{B\L} or $\textbf{B\L}_x$, where $x$ is a string that is a $k$-combination of \textbf{D}, \textbf{4}, \textbf{5}, and \textbf{T} ($k\in \{1,\ldots,4\}$). A derivation in \textbf{A} is a sequence of formulae like $\varphi_1, \cdots, \varphi_n=\varphi$ such that for all  $i\in \{1,\cdots,n\}$, either $\varphi_i$ is an axiom of \textbf{A} or it is obtained from some $\varphi_j$s  using (R$_\text{MP}$) or (R$_\text{GB}$), where $j\leq i$. Then we say that $\varphi$ is provable in \textbf{A} and denote it by $\vdash_\textbf{A}\varphi$.

Also, if $\Gamma$ is a set of formulae, and there is a sequence $\psi_1,\cdots, \psi_m=\varphi$ such that each $\psi_i$ is an axiom of \textbf{A} or is a member of $\Gamma$ or is obtained by applying (R$_{MP}$) and (R$_{GB}$) with this restriction that applying (R$_{GB}$) is not allowed for non-valid formulae $\psi_j\in \Gamma$, then we say that $\varphi$‌ is provable from $\Gamma$, denoted as $\Gamma \vdash_{\textbf{A}} \varphi$. It's important to note that we refrain from applying (R$_{GB}$) to non-valid formulas, as our goal is to avoid endorsing or believing on non-valid formulas.



\end{definition}

\begin{theorem}\label{soundness_pseudo-classical} \textbf{(Soundness)}
Let $\mathcal{M}$ be a class of D\L L-models and $x$ be a string which is a $k$-combination of \textbf{D}, \textbf{4}, \textbf{5}, and \textbf{T}, where $k \in \{1,\ldots, 4\}$.
\begin{enumerate}
\item If $\; \vdash_\textbf{B\L} \varphi$, then $\mathcal{M}\vDash \varphi$.
\item If  $\; \vdash_{\textbf{B\L}_\textbf{D}} \varphi$ and $\mathcal{M}$ is serial, then $\mathcal{M}\vDash \varphi$.
\item If  $\; \vdash_{\textbf{B\L}_\textbf{4}} \varphi$ and $\mathcal{M}$ is transitive, then $\mathcal{M}\vDash \varphi$.
\item If  $\; \vdash_{\textbf{B\L}_\textbf{5}} \varphi$ and $\mathcal{M}$ is recognizable, then $\mathcal{M}\vDash \varphi$.
\item If  $\; \vdash_{\textbf{B\L}_\textbf{T}} \varphi$ and $\mathcal{M}$ is reflexive, then $\mathcal{M}\vDash \varphi$.
\item If $\; \vdash_{\textbf{B\L}_{x}} \varphi$ and $\mathcal{M}$ has the properties of accessibility relations corresponding to the axioms of $x$, then $\mathcal{M} \vDash \varphi$.
\end{enumerate}
\end{theorem} 

\begin{proof}
The proof is obtained straightforwardly by applying Proposition \ref{distribution_ax},  Lemma \ref{gb_structures_properties}, and Lemma \ref{gb_admissible_rules}.
\end{proof}


	\begin{definition}
	Let \textbf{A} be an axiomatic system. We say that a finite set $\{\varphi_1, \cdots, \varphi_n\}$  of D\L L-formulae is $\textbf{A}$-consistent or simply \textit{consistent}, if $\nvdash_{\textbf{A}} \neg (\varphi_1\,\&\, \cdots\,\&\, \varphi_n)$. If $\Gamma$ is an infinite set of D\L L-formulae and all its finite subsets are $\textbf{A}$-consistent, then $\Gamma$ is called $\textbf{A}$-consistent.
		
	An infinite consistent set  $\Phi$ of D\L L-formulae  is called \textit{maximal} if for all formulae  $\psi\notin\Phi$, the set $\Phi \cup \{\psi\}$  is not consistent.	
	An infinite set $\Phi$ of D\L L-formulae is called \textit{maximal and \textbf{A}-consistent} if the following conditions hold:
	\begin{enumerate}
		\item The set $\Phi$ is $\textbf{A}$-consistent.
		\item For all D\L L-formula $\psi\notin\Phi$, the set $\Phi \cup \{\psi\}$  is not \textbf{A}-consistent.
	\end{enumerate}
	
\end{definition}

%
%

\begin{remark}
	
	
	While we have incorporated the notion of maximality, it's essential to clarify that, unlike classical logic, here a maximal set does not necessarily encompass all formulas or their negations. It's conceivable that there may exist a formula for which neither the formula itself nor its negation is included in any maximal set. To illustrate this point, consider the expansion of the consistent set $\{\neg(\varphi \,\&\, \varphi), \neg(\neg \varphi \,\&\, \neg \varphi)\}$ to a maximal set using the method outlined in Lemma \ref{maximal}-i. Interestingly, the resulted set does not contain either $\varphi$ or $\neg \varphi$.

\end{remark}

\begin{lemma}\label{lemma_consistent}
Let \textbf{A} be an axiomatic system and $\Phi$ be a set of D\L L-formulae that is not $\textbf{A}$-consistent. For each set $\Psi$ which $\Phi \subseteq \Psi$  we have $\Psi$ is not  $\textbf{A}$-consistent.
\end{lemma}

\begin{proof}
	It can be easily checked.
\end{proof}
\begin{lemma}\label{and_Luka}
	Let $\Gamma = \{\varphi_1, \cdots, \varphi_n\}$ be a set of formulae, then $\Gamma \vdash \varphi_1\,\&\,\cdots\,\&\,\varphi_n$. 
\end{lemma}
\begin{proof}

	Without loss of generality, we show that if $\Gamma = \{\varphi_1, \varphi_2\}$, then  $\Gamma \vdash\varphi_1\,\&\,\varphi_2$.
	\begin{align*}
		&(1)&& \Gamma\vdash (\varphi_2\,\&\, \varphi_1)\rightarrow (\varphi_1\,\&\, \varphi_2)&& (A3)\\
		&(2)&& \Gamma\vdash \varphi_2\rightarrow(\varphi_1\rightarrow (\varphi_1\,\&\, \varphi_2)) && \text{(1),(A5), MP}\\
		&(3)&& \Gamma\vdash\varphi_2&& \text{assumption } \varphi_2 \in \Gamma \\
		&(4)&& \Gamma\vdash\varphi_1&& \text{assumption } \varphi_1 \in \Gamma \\
		&(5)&& \Gamma\vdash\varphi_1\rightarrow (\varphi_1\,\&\, \varphi_2)&& \text{(2),(3),MP}\\
		&(6)&&\Gamma\vdash\varphi_1\,\&\, \varphi_2 && \text{(4),(5),MP}
	\end{align*}

\end{proof}

\begin{lemma}\label{maximal}
Let $\textbf{A}$ be an axiomatic system. 
\begin{itemize}
\item[(i)] Each $\textbf{A}$-consistent set $\Phi$ of D\L L-formulae can be extended to a maximal $\textbf{A}$-consistent set.
\item[(ii)] If $\Phi$ is a maximal  $\textbf{A}$-consistent set, then for all D\L L-formulae $\varphi$ and $\psi$:
\begin{enumerate}
	\item $\varphi\,\&\,\psi \in \Phi$ if and only if $\varphi \in \Phi$ and $\psi\in \Phi$,
	\item If $\varphi \in \Phi$ and $\varphi \rightarrow \psi \in \Phi$, then $\psi \in \Phi$,
	\item If $\Phi\vdash_{ } \varphi$, then $\varphi \in \Phi$,
	\item $\exists n \in \mathbb{N}\; \underbrace{\varphi\,\&\,\cdots\,\&\,\varphi}_{\text{n times}} \in \Phi$ or $\neg \underbrace{(\varphi\,\&\,\cdots\,\&\,\varphi)}_{\text{n times}} \in \Phi$. 
\end{enumerate}
\end{itemize}
\end{lemma}

\begin{proof}
\textbf{$(i)$:}	Let $\varphi_1, \varphi_2, \cdots$
 be an enumeration of all D\L L-formulae. We define a sequence of sets of formulas $\Phi =\Phi_1\subseteq \Phi_2 \subseteq \cdots \subseteq \Phi_i \subseteq \cdots$  as follows:
\begin{equation*}
	\forall i~~~~\Phi_{i+1} = \left\lbrace \begin{array}{lc}
		\Phi_i & \text{If } \Phi_i\cup \{\varphi_i\} \text{ is not } \textbf{A}\text{-consistent},\\
		\Phi_i \cup \{\varphi_i\} & \text{otherwise}.
	\end{array} \right. 
\end{equation*}
Let $\Phi_\omega = \bigcup_{i\geq 1} \Phi_i$. It is easy to see that $\Phi_\omega$ is \textbf{A}-consistent since otherwise there is a finite subset $\Gamma\subset \Phi_\omega$ such that $\Gamma$ is inconsistent, meanwhile for some $j \in \mathbb{N}$ we have $\Gamma = \Phi_j$ which is a \textbf{A}-consistent set, therefore $\Phi_\omega$ is $\textbf{A}$-consistent. Now suppose that $\Phi_\omega$ is not maximal, so there exists a formula $\varphi=\varphi_j$ such that $\varphi \notin \Phi_\omega$ and the set $\Phi_\omega \cup\{\varphi\}$ is a \textbf{A}-consistent set. 
If $\Phi_j\cup \{\varphi_j\}$ is not \textbf{A}-consistent,  then $\Phi_\omega \cup \{\varphi_j\}$ is not \textbf{A}-consistent by Lemma \ref{lemma_consistent} which is a contradiction. Otherwise $\Phi_{j+1}$  contains $\varphi_j$, and so  $\Phi_\omega$ contains $\varphi_j$ which has a contraction to the assumption $\varphi \notin \Phi_\omega$.
\vskip 0.3cm
\textbf{(ii)-1:}  Let $\varphi \, \&\, \psi \in \Phi$ and $\varphi \notin \Phi$. Thus $\Phi \cup \{\varphi\}$ is not \textbf{A}-consistent, by definition. So there is a finite set $\Gamma\subset \Phi$ such that $\Gamma \cup \{\varphi\}$ is not \textbf{A}-consistent. Let $\Gamma = \{\psi_1, \cdots, \psi_n\}$, hence by the inconsistency of $\Gamma \cup \{\varphi\}$ we have $\vdash_{\textbf{A}} \neg (\psi_1 \,\&\, \cdots \,\&\, \psi_n \,\&\, \varphi)$. Since $\Gamma\cup \{\varphi, \psi\}$ is not $\textbf{A}$-consistent, then we  have $\vdash_{\textbf{A}} \neg (\psi_1 \,\&\, \cdots \,\&\, \psi_n \,\&\, \varphi \,\&\, \psi)$, and since $\Gamma \cup \{\varphi\,\&\,\psi\}\subset \Phi$, then using Lemma \ref{lemma_consistent}, we obtain a contradiction to the  \textbf{A}-consistency of $\Phi$.

For other direction assume $\varphi\in \Phi$, $\psi \in \Phi$ and $\varphi \,\&\, \psi \notin \Phi$. Then $\Phi \cup \{\varphi\,\&\, \psi\} $ is not \textbf{A}-consistent, and there is a finite set  $\Gamma = \{\psi_1, \cdots, \psi_n\} \subset \Phi$ such that $\vdash_{\textbf{A}}\neg (\psi_1\,\&\, \cdots\,\& \,\psi_n\,\&\, \varphi \,\&\,\psi)$. But this states that the finite subset $\Gamma \cup \{\varphi,\psi\} \subset \Phi$ is not $\textbf{A}$-consistent. 
\\
\textbf{(ii)-2:}
Assume $\varphi \in \Phi$ and $\varphi \rightarrow \psi \in \Phi$. If $\psi \notin \Phi$, then $\Phi \cup \{\psi\}$ is not \textbf{A}-consistent and there exists a finite subset $\Gamma \subset \Phi$ such that $\Gamma \cup \{\psi\}$ is not \textbf{A}-consistent. Let $\Gamma = \{\psi_1, \cdots, \psi_n\}$, so we have:
\begin{equation}\label{eq_inconsistent}
	\vdash_{\textbf{A}} \neg (\psi_1\,\&\, \cdots, \,\&\, \psi_n \,\&\, \psi).
\end{equation} 
On the other hand,
it can be proved that from $\vdash_{\textbf{A}} \chi\rightarrow \chi$ and 
(\L9) we have 
$\vdash_{\textbf{A}}(\chi \rightarrow \chi) \,\&\, ((\varphi \,\&\, (\varphi \rightarrow \psi)) \rightarrow \psi)$
and by considering an instance of (\L10) and applying ($R_{MP}$) we obtain
$$\vdash_{\textbf{A}}\chi \,\&\, (\varphi \,\&\, (\varphi \rightarrow \psi)) \rightarrow (\chi \,\&\, \psi).$$
Then by replacing $\psi_1\,\&\, \cdots \,\&\,\psi_n$ instead of $\chi$ we have:
$$\vdash_{\textbf{A}}(\psi_1\,\&\,\cdots\,\&\, \psi_n \,\&\, \varphi \,\&\,(\varphi \rightarrow \psi)) \rightarrow (\psi_1 \,\&\, \cdots \,\&\, \psi_n \,\& \,\psi )$$
Now, by using $(\varphi \rightarrow \psi)\rightarrow (\neg \psi \rightarrow \neg \varphi)$ scheme, which is provable in \L ukasiewicz logic we have:
\begin{equation}\label{eq013} 
	\vdash_{\textbf{A}}\neg (\psi_1 \,\&\, \cdots \,\&\, \psi_n \& \psi ) \rightarrow \neg (\psi_1\,\&\,\cdots\,\&\, \psi_n \,\&\, \varphi \,\&\,(\varphi \rightarrow \psi)).
\end{equation}
By applying (R$_{\text{MP}}$) on \ref{eq_inconsistent}, \ref{eq013} we obtain $\vdash_{\textbf{A}} \neg (\psi_1\,\&\,\cdots\,\&\, \psi_n \,\&\, \varphi \,\&\,(\varphi \rightarrow \psi)) $. But this means that $\Gamma \cup \{\varphi, \varphi\rightarrow \psi\}\subset \Phi$ is inconsistent, which is a contradiction to the \textbf{A}-consistency of $\Phi$.\\
\textbf{(ii)-3:}  
We first show that if $\vdash \varphi$, then $\varphi \in \Phi$.
Let $\vdash \varphi$. 
First, note that we have $\vdash\neg  \varphi \leftrightarrow \perp$ since from the assumption we have $\vdash \neg\neg \varphi$ from (\L 11) and  (R$_{\text{MP}}$), then using (\L 15) we obtain $\vdash \neg \varphi \rightarrow \perp$. Also, from (A7) we have $\vdash \perp \rightarrow \neg \varphi$.

Now, for the sake of contradiction assume that $\varphi \notin \Phi$. If $\varphi \notin \Phi$, then $\Phi\cup \{\varphi\}$ is inconsistent, by definition of a maximal set. So there is a subset $\Gamma = \{\psi_1,\cdots, \psi_n\} \subset \Phi$, such that $\Gamma\cup \{\varphi\}$ is not consistent. Therefore 
$$\begin{array}{llll}
	(1) & & \vdash_{ }\neg (\psi_1\,\&\, \cdots \,\&\, \psi_n \,\&\, \varphi) & \text{inconsistency of } \Gamma \cup \{\varphi\}\\
	(2) & & \vdash_{ }\neg \psi_1 \veebar\cdots \veebar \neg\psi_n \veebar \neg \varphi & (1),(\L5), \text{MP} \\
	(3) & & \vdash_{ }\neg (\neg \psi_1 \veebar\cdots \veebar \neg\psi_n) \rightarrow \neg \varphi & (2), (\L7 ), \text{MP}\\
	(4) & & \vdash_{ }\neg (\neg \psi_1 \veebar\cdots \veebar \neg\psi_n) \rightarrow \perp & (3), \text{assumption and replacement, MP} \\
	(5) & & \vdash_{ }(\neg \neg \psi_1 \,\&\, \cdots \,\&\, \neg \neg \psi_n) \rightarrow \perp & (4), (\L6), \text{MP} \\
	(6)  & & \vdash_{ }(\psi_1 \,\&\, \cdots \,\&\, \psi_n) \rightarrow \perp & (5), (\L8), \text{replacement}, \text{MP}\\
	(7) & & \vdash_{ }\neg (\psi_1 \,\&\, \cdots \,\&\,\psi_n) & (6),(\L9), \text{MP}
\end{array}$$
Thus from $\vdash_{ } \neg (\psi_1 \,\&\, \cdots \,\&\,\psi_n)$, it follows that $\Gamma$ is not consistent. However, this contradicts the assumption of the consistency of $\Phi$.
By applying induction based on the length of the proof of $\varphi$ from $\Phi$, along with the previous observation and part (ii)-2, we can establish the desired statement. 

\textbf{(ii)-4:} Suppose that there is a D\L L-formula $\varphi$ such that for every $n\in \mathbb{N}$  $\underbrace{\varphi\,\&\,\cdots\,\&\,\varphi}_{\text{n times}}\notin \Phi$ and $\neg \underbrace{(\varphi\,\&\,\cdots\,\&\,\varphi)}_{\text{n times}} \notin \Phi$. From the maximality and consistency of $\Phi$, for all $n\in\mathbb{N}$ neither  $\Phi \cup \{\underbrace{\varphi\,\&\,\cdots\,\&\,\varphi}_{\text{n times}}\}$ nor $\Phi \cup \{\neg \underbrace{(\varphi\,\&\,\cdots\,\&\,\varphi)}_{\text{n times}}\}$ is \textbf{A}-consistent.
Assume  $\Gamma = \{\psi_1, \cdots, \psi_m\}$ is a finite subset of $\Phi$ such that   $\Gamma \cup \{\underbrace{\varphi\,\&\,\cdots\,\&\,\varphi}_{\text{n times}}\}$ is not \textbf{A}-consistent,
then we have the following deduction:
$$\begin{array}{llll}
	(1)& &\vdash_{\textbf{A}} \neg (\psi_1 \,\&\, \cdots \,\&\, \psi_m \,\&\,\underbrace{(\varphi\,\&\,\cdots\,\&\,\varphi)}_{\text{n times}}) & \text{inconsistency of } \Gamma\cup \{\underbrace{\varphi\,\&\,\cdots\,\&\,\varphi}_{\text{n times}}\}\\
	(2)& & \vdash_{\textbf{A}} \neg \psi_1 \veebar \cdots \veebar \neg \psi_m \veebar \neg \underbrace{(\varphi\,\&\,\cdots\,\&\,\varphi)}_{\text{n times}} & (1), (\text{\L} 5), (\text{R}_{\text{MP}}), (\text{R}_{\text{MP}})\\ 
	(3)& &\vdash_{\textbf{A}} \neg (\neg \psi_1 \veebar \cdots \veebar \neg \psi_m ) \rightarrow \neg \underbrace{(\varphi\,\&\,\cdots\,\&\,\varphi)}_{\text{n times}} & (2), (\text{\L}7) , (\text{R}_{\text{MP}})\\
	(4)& &\vdash_{\textbf{A}} (\neg\neg \psi_1 \,\&\, \cdots \,\&\,\neg\neg\psi_m ) \rightarrow \neg \underbrace{(\varphi\,\&\,\cdots\,\&\,\varphi)}_{\text{n times}}  & (3), (\text{\L}6), (\text{R}_{\text{MP}}) \\
	(5)& & \vdash_{\textbf{A}} (\psi_1 \,\&\, \cdots \,\&\,\psi_m ) \rightarrow \neg \underbrace{(\varphi\,\&\,\cdots\,\&\,\varphi)}_{\text{n times}} & (4), (\text{\L}11), (\text{R}_{\text{MP}}).
\end{array}$$
Thus from (ii)-3 we have $(\psi_1 \,\&\, \cdots \,\&\, \psi_m )\rightarrow \neg \underbrace{(\varphi\,\&\,\cdots\,\&\,\varphi)}_{\text{n times}} \in \Phi$, and
since $\Gamma=\{\psi_1, \cdots, \psi_m\}\subset \Phi$, then we have $(\psi _1 \,\&\, \cdots \,\&\, \psi_m) \in \Phi$, by (ii)-1. So from (ii)-2 we have $\neg \underbrace{(\varphi\,\&\,\cdots\,\&\,\varphi)}_{\text{n times}} \in \Phi$, which is a  contradiction with the assumption  $\neg \underbrace{(\varphi\,\&\,\cdots\,\&\,\varphi)}_{\text{n times}} \notin \Phi$.
\end{proof}

\begin{theorem}\label{consistent_e}
	Let \textbf{A} be an axiomatic system, $\Phi$ be an  \textbf{A}-consistent set of D\L L formulae and $\varphi$ be a D\L L-formula such that 
	$\Phi \nvdash_{\textbf{A}} \varphi$. If $\Phi^{*} = \Phi\cup \{\neg \varphi\}$, then $\Phi^{*}$ is \textbf{A}-consistent.
\end{theorem}
\begin{proof}
	For the sake of contradiction suppose that $\Phi^{*}$ is not \textbf{A}-consistent. So there exists a finite subset $\Gamma = \{\varphi_1, \cdots, \varphi_n\}\subseteq\Phi$ such that $\Gamma \cup \{\neg \varphi\}$ is not \textbf{A}-consistent, then we have:
	\begin{align}
		\nonumber	& \vdash_{\textbf{A}} \neg (\varphi_1\,\&\, \cdots \,\&\, \varphi_n \,\&\, \neg \varphi)  &  \\
		\nonumber		& \vdash_{\textbf{A}} \neg \varphi_1 \veebar \cdots \veebar \neg \varphi_n \veebar \neg \neg \varphi & \\
		\nonumber		& \vdash_{\textbf{A}} \neg (\neg \varphi_1 \veebar \cdots \veebar \neg \varphi_n) \rightarrow \varphi & \\
		\label{eq_008}		& \vdash_{\textbf{A}} (\varphi_1 \,\&\, \cdots \,\&\, \varphi_n) \rightarrow \varphi &
	\end{align}
	Also, by Lemma \ref{and_Luka}, from $\Gamma\subset \Phi$ we have $\Phi \vdash_{\textbf{A}} (\varphi_1 \,\&\, \cdots \,\&\, \varphi_n) $.
	Thus from \ref{eq_008}, we have $\Phi \vdash_{\textbf{A}} \varphi$,
	which is a contradiction to $\Phi \nvdash_{\textbf{A}} \varphi$.
	
\end{proof}

\begin{theorem} \label{model_e}
Let $\Phi$ be a \textbf{B\L}-consistent set and $\Phi \vdash_{\textbf{B\L}} \varphi$, where $\varphi$ is a D\L L-formula. Then there is a D\L L-model $\mathfrak{M}$ and a state $s$ such that $V_{s}^{\mathfrak{M}}(\varphi) = 1$. 
\end{theorem}

\begin{proof}
	The desired model $\mathfrak{M}= (S , r  , \pi )$  is defined as follows:
	\begin{align*}
		& S  = \{ s_{\Phi}\,|\, \Phi \text{ is a maximal and consistent set of formulae}\} \\
		& r   (s_\Phi , s_\Psi) = \left\lbrace \begin{array}{lc}
			1 & \Phi\backslash B  \subseteq \Psi \\
			0 & \text{otherwise}
		\end{array}\right.; \qquad \Phi \backslash B  \stackrel{def}{=} \{\varphi \,|\, B  \varphi \in \Phi\}~\text{and}~s_\Phi ,~ s_\Psi \in S\\
		& \pi  (s_{\Phi}, p) = \left\lbrace \begin{array}{lc}
			1 & p\in \Phi \\
			0 & \neg p \in \Phi \\
			0.5 & \text{otherwise}
		\end{array} \right.; \qquad p\in \mathcal{P}.
	\end{align*}
	%

Let $\varphi$ be a D\L L-formula. By induction on the complexity of $\varphi$ we prove that for each maximal and \textbf{B\L}-consistent set  $\Phi^*$:
\begin{equation} \label{eq014}
	V_{s_{\Phi^*}}^{\mathfrak{M}}(\varphi) = \left\lbrace 
	\begin{array}{lll}
		1 &  & \varphi\in \Phi^* \\
		0 &  & \neg \varphi\in \Phi^* \\
		0.5 &  & \text{otherwise}
	\end{array}
	\right.
\end{equation}


where $s_{\Phi^*}\in S $.
Note that \textbf{B\L}-consistent set $\Phi$ has a maximal and \textbf{B\L}-consistent extension $\Psi$ using Lemma \ref{maximal}. Consequently, it follows that $V_{s_{\Psi}}^{\mathfrak{M}}(\varphi)=1$ according to equation \ref{eq014}, which aligns with the intended assertion.

Let $\Phi^*$ be a maximal and \textbf{B\L}-consistent set. 
The base step is obvious by the definition of $\pi $.  For the induction step, we have the following cases:
\\
\textbf{case 1:} $\varphi = \neg \psi$. 
If $\neg \psi \in \Phi^*$, then by induction on $\psi$ we have $V_{s_{\Phi^*}}^{\mathfrak{M}}(\psi) =0$, and so by definition $V_{s_{\Phi^*}}^{\mathfrak{M}}(\neg \psi)=1$.

Now assume $\neg \psi \notin \Phi^*$. By induction on $\psi$ we obtain $V_{s_{\Phi^*}}^{\mathfrak{M}}(\psi) =1$, and $V_{s_{\Phi^*}}^{\mathfrak{M}}(\psi) =0.5$, if $\psi \in \Phi^*$ and $\psi \notin \Phi^*$, respectively.
For the other side we have
\begin{equation}\label{eq_000}
	V_{s_{\Phi^*}}^{\mathfrak{M}}(\neg \psi)=1 \stackrel{\text{def.}}{\Rightarrow}  V_{s_{\Phi^*}}^{\mathfrak{M}}(\psi) =0 \stackrel{\text{Induction}}{\Rightarrow}  \neg\psi \in \Phi^*.
\end{equation}\\
\textbf{case 2:} $\varphi = \psi \,\&\, \chi$. We have:
$$\psi \,\&\, \chi \in \Phi^* \stackrel{\text{Lemma } \ref{maximal}}{\iff} \psi \in \Phi^*, \chi \in \Phi^*\stackrel{\text{induction}}{\iff} V_{s_{\Phi^*}}^{\mathfrak{M}} (\psi)=1= V_{s_{\Phi^*}}^{\mathfrak{M}}(\chi) \Leftrightarrow V_{s_{\Phi^*}}^{\mathfrak{M}}(\psi \,\&\, \chi)=1. $$
The case $\psi\,\&\,\chi \notin \Phi^*$ has a similar argument.\\
\textbf{case 3:} $\varphi = \psi \rightarrow \chi$.
First assume  $V_{s_{\Phi^*}}^{\mathfrak{M}}(\psi\rightarrow \chi)=1$.
We have six possible cases in which $V_{s_{\Phi^*}}^{\mathfrak{M}}(\psi) \leq V_{s_{\Phi^*}}^{\mathfrak{M}}(\chi)$. If $V_{s_{\Phi^*}}^{\mathfrak{M}}(\chi)=1$, then by induction hypothesis we have $\chi \in \Phi^*$, and so using instances of (\L1) and Lemma \ref{maximal} we obtain $\psi\rightarrow\chi \in \Phi^*$. Similarly, if
 $V_{s_{\Phi^*}}^{\mathfrak{M}}(\psi)=0$, then by induction hypothesis we have $\neg \psi\in \Phi^*$, then using instances of (\L1), (\L3) and Lemma \ref{maximal} we obtain $\psi\rightarrow\chi \in \Phi^*$. The only remaining case is when $V_{s_{\Phi^*}}^{\mathfrak{M}}(\psi) = V_{s_{\Phi^*}}^{\mathfrak{M}}(\chi)=0.5$. By induction we have $\psi,\neg \chi \notin \Phi^*$, so using Lemma  \ref{maximal} we have $\psi\,\&\,\neg\chi \notin \Phi^*$, and it can be easily obtained that $\neg (\neg \psi \rightarrow \chi) \notin \Phi^*$. By applying Theorem \ref{consistent_e} the set $\Phi^*\cup \{\neg\neg (\psi\rightarrow\chi)\}$ is consistent, and by the maximality of  $\Phi^*$ we have $\neg\neg (\psi\rightarrow\chi) \in \Phi^*$. The desired statement can be derived using  (\L10) and Lemma \ref{maximal}.

For the other direction, suppose that $\psi\rightarrow\chi \in \Phi^*$. If $\chi \in \Phi^*$, by induction hypothesis the desired statement can be obtained easily. If $\psi\in \Phi^*$, then using Lemma \ref{maximal} we obtain $\chi \in \Phi^*$ and again the desired statement can be concluded. So assume $\chi \notin \Phi^*$ and $\neg \psi\in \Phi^*$, then by induction hypothesis on $\psi$ we have $V_{s_{\Phi^*}}^{\mathfrak{M}}(\psi)=0$,  thus $V_{s_{\Phi^*}}^{\mathfrak{M}}(\psi\rightarrow\chi)=1$ by definition. Now, let $\neg \psi\notin \Phi^*$. If $\neg \chi \notin \Phi^*$, then by induction we have $V_{s_{\Phi^*}}^{\mathfrak{M}}(\psi)=V_{s_{\Phi^*}}^{\mathfrak{M}}(\chi)=0.5$ that concludes the desired statement. If $\neg \chi \in \Phi^*$, then using (\L3) and Lemma \ref{maximal} we will have $\neg \psi\in \Phi^*$ which is contradiction with our assumption.

If $\psi\rightarrow\chi\notin \Phi^*$, then $\Phi^*\cup \{\psi\rightarrow\chi\}$ is inconsistent.
So there is a finite set $\Gamma = \{\psi_1,\cdots,\psi_n\}\subset \Phi^*$ such that $\Gamma \cup \{\psi\rightarrow\chi\}$ is inconsistent. Thus we have
$$\begin{array}{lll}
	(1) & \vdash_{\textbf{B\L}} \neg ( \psi_1\,\&\, \cdots \,\&\, \psi_n\,\&\,(\psi\rightarrow\chi)) & \text{Assumption}\\
(2) & \vdash_{\textbf{B\L}} \neg \psi_1 \veebar \cdots \veebar \neg \psi_n \veebar \neg (\psi\rightarrow\chi) & \text{1, (\L5)}\\
(3) & \vdash_{\textbf{B\L}} \neg (\neg  \psi_1 \veebar \cdots \veebar \neg \psi_n ) \rightarrow \neg (\psi\rightarrow\chi) & \text{2, (\L7)}\\
(4) & \vdash_{\textbf{B\L}} (\neg\neg \psi_1 \,\&\,\cdot\,\&\, \neg\neg \psi_n)\rightarrow \neg (\psi\rightarrow\chi) & \text{3, (\L5)}\\
(5) & \vdash_{\textbf{B\L}} ( \psi_1 \,\&\,\cdot\,\&\, \psi_n)\rightarrow \neg (\psi\rightarrow\chi) & \text{4, (\L11), replacement}\\
(6) & \vdash_{\textbf{B\L}} ( \psi_1 \,\&\,\cdot\,\&\, \psi_n)\rightarrow (\psi \,\&\,\neg \chi)& \text{5, (\L6), (\L 7), replacement}\\
\end{array}$$

So since $\psi_1\,\&\,\cdots\,\&\,\psi_n \in \Phi^*$, considering (6) above and using Lemma \ref{maximal} we conclude that  $\psi \,\&\,\neg \chi\in \Phi^*$. Applying Lemma \ref{maximal}-i we can deduce that $\psi \in \Phi^*$ and $\neg \chi \in \Phi^*$. By induction, we can therefore establish  $V_{s_{\Phi^*}}^{\mathfrak{M}}(\psi)=1$ and $V_{s_{\Phi^*}}^{\mathfrak{M}}(\chi)=0$. Thus $V_{s_{\Phi^*}}^{\mathfrak{M}}(\psi\rightarrow\chi)=0$ as desired.\\
\\
\textbf{case 4:} $\varphi = B  \psi$.  
First assume $V_{s_{\Phi^*}}^{\mathfrak{M}}(B  \psi) = 1$, so $\inf_{s'\in S }\max\{1-r (s_{\Phi^*}, s'), V_{s' }^{\mathfrak{M}} (\psi)\} = 1$, by definition.
Thus for all $s'\in S  $ we have $\max\{1-r (s_{\Phi^*}, s'), V_{s' }^{\mathfrak{M}} (\psi)\}=1$. 
\\
\textbf{Claim:} $\Phi^*\backslash B  \cup \{\neg \psi\}$ is not \textbf{B\L}-consistent.
\\
\textit{Proof of the claim:} For the sake of contradiction suppose that $\Phi^*\backslash B  \cup \{\neg \psi\}$ is \textbf{B\L}-consistent. Then by lemma \ref{maximal},
we can establish the existence of a maximal and \textbf{B\L}-consistent extension $\Psi$ for the set 
  $\Phi^*\backslash B  \cup \{\neg \psi\}$.
  Since $\Phi^*\backslash B  \cup \{\neg \psi\}\subset \Psi$ we have $r (s_{\Phi^*}, s_{\Psi})=1$, by definition of the model. Furthermore $\neg \psi \in \Psi$, thus by induction hypothesis we have $V_{s_{\Psi}}^{\mathfrak{M}}(\psi) = 0$. Since $\max \{1-r (s_{\Phi^*}, s_{\Psi}), V_{s_{\Psi}}^{\mathfrak{M}}(\psi)\}=0$ we have a contradiction with our assumption in case 4.
\hfill $\blacksquare_{\text{Claim}}$

Therefore, according to the Claim, there exists a finite subset $\Gamma = \{\psi_1 , \cdots, \psi_n, \neg \psi\} \subset \Phi^*\backslash B  \cup \{\neg \psi\}$ which is not \textbf{B\L}-consistent.  So we have:
$$\begin{array}{llll} 
	(1)& &\vdash_{\textbf{B\L}} \neg (\psi_1 \,\&\, \cdots \,\&\, \psi_n \,\&\,\neg \psi) & \text{inconsistency of } \Gamma\\
	(2)& & \vdash_{\textbf{B\L}} \neg \psi_1 \veebar \cdots \veebar \neg \psi_n \veebar \neg \neg \psi & (1), (\text{\L} 5)\\ 
	(3)& &\vdash_{\textbf{B\L}} \neg (\neg \psi_1)\rightarrow(\neg \psi_2 \veebar \cdots \veebar \neg \psi_n \veebar \neg \neg \psi) & (2), (\text{\L}7) \\
	(4)& &\vdash_{\textbf{B\L}} \psi_1\rightarrow(\neg \psi_2 \veebar \cdots \veebar \neg \psi_n \veebar \psi)  & (3), \text{replacement}\\
	(5)& & \vdots & \\
	(6)& & \vdash_{\textbf{B\L}} \psi_1 \rightarrow (\psi_2 \rightarrow \cdots (\psi_n\rightarrow \psi))  & \text{similar procedure for all } \psi_i s\\
	(7)& &\vdash_{\textbf{B\L}} B ( \psi_1 \rightarrow (\psi_2 \rightarrow \cdots (\psi_n\rightarrow \psi))) & (6), (R_{GB})\\
	(8)& &\multicolumn{2}{l}{\vdash_{\textbf{B\L}} (B  \psi_1 \,\&\, B ( \psi_1 \rightarrow (\psi_2 \rightarrow \cdots (\psi_n\rightarrow \psi))))\rightarrow (B (\psi_2 \rightarrow \cdots (\psi_n\rightarrow \psi)))} \\
	& & &\text{instances of }(\text{\L}_B1) \\
\end{array}$$
From (7) above we have $B ( \psi_1 \rightarrow (\psi_2 \rightarrow \cdots (\psi_n\rightarrow \psi))) \in \Phi^*$, by applying Lemma \ref{maximal}. Meanwhile since $\psi_1, \cdots, \psi_n \in \Phi^*\backslash B  $, we have $B  \psi_1, \cdots, B  \psi_n \in \Phi^*$. So by (8) and statements
$B  \psi_1 \in \Phi^*$ and $B ( \psi_1 \rightarrow (\psi_2 \rightarrow \cdots (\psi_n\rightarrow \psi))) \in \Phi^*$ we have:
$$B (\psi_2 \rightarrow \cdots (\psi_n\rightarrow \psi))\in\Phi^*,$$
using Lemma \ref{maximal}.
By similar argument, finally we obtain $B  \psi \in \Phi^*$.


For the other direction, assume $B  \psi \in \Phi^*$, thus $\psi \in \Phi^*\backslash B $ by definition. Let $\Psi$ be an arbitrary maximal and consistent set. If $\Phi^*\backslash B  \subset \Psi$, then $r  (s_{\Phi^*}, s_{\Psi})=1$, and so since $\psi \in \Psi$, then by induction hypothesis we have $V_{s_{\Psi}}^{\mathfrak{M}}(\psi)=1$. Therefore $\max\{1-r  (s_{\Phi^*}, s_{\Psi}), V_{s_{\Psi}}^{\mathfrak{M}}(\psi)\}=1$. Otherwise, if $\Phi^*\backslash B  \not\subset \Psi$ then $r (s_{\Phi^*}, s_{\Psi})=0$, and again we have $\max\{1-r  (s_{\Phi^*}, s_{\Psi}), V_{s_{\Psi}}^{\mathfrak{M}}(\psi)\}=1$. Thus we have that $V_{s_{\Phi^*}}^{\mathfrak{M}}(B  \psi)=1$, by definition.	

Now assume $B\psi \notin \Phi^*$, so by definition $\psi\notin \Phi^*\backslash B$. Meanwhile $\Phi^*\backslash B$ is consistent.
	Note that if $\Phi^*\backslash B$ is inconsistent, then there is a finite set  $\Gamma = \{\psi_1, \cdots, \psi_n \}\subset \Phi^*\backslash B$ such that $\Gamma $ is inconsistent. By Lemma \ref{lemma_consistent}, $\Gamma' = \{\psi_1, \cdots, \psi_n, B \psi'\}$ is inconsistent, where $\psi '$ is a formula such that $\vdash \psi'$ and $B\psi' \in \Phi^*$. So we have the following argument:	
	$$\begin{array}{l}
		\vdash_{\textbf{B\L}} \neg (\psi_1\,\&\, \cdots\,\&\,\psi_n \,\&\, B \psi'), \\
		\vdash_{\textbf{B\L}} \neg \psi_1 \veebar \cdots \veebar \neg \psi_n \veebar \neg B\psi', \\
		\vdash_{\textbf{B\L}} \neg (\neg \psi_1\veebar \cdots\veebar\neg \psi_n)\rightarrow \neg B \psi',\\
		\vdash_{\textbf{B\L}} (\psi_1\,\&\, \cdots\,\&\,\psi_n)\rightarrow \neg  B \psi' \\ 
			\end{array}$$
Therefore $\Gamma \vdash_{\textbf{B\L}} \neg B \psi'$, and  this results the inconsistency of $\Phi^*$ which is a contradiction.
 Thus there is a maximal and consistent set $\Psi$ containing $\Phi^*\backslash B$. It is enough to consider the set $\Psi$ in which $\psi\notin \Psi$. Hence, $\max\{1-r (s_{\Phi^*}, s_\Psi), V_{s_{\Psi}}(\psi)\}\in \{0, 0.5\}$ and so $V_{s_{\Phi^*}}^{\mathfrak{M}}(B\psi)\in\{0, 0.5\}$.

\end{proof}
\begin{corollary}\label{corollary_model_e}
The following statements hold:
\begin{itemize}
	\item[(i)] 	Let $\Phi$ be a $\textbf{B\L}_\textbf{D}$-consistent set and $\Phi \vdash_{\textbf{B\L}_\textbf{D}}\varphi$. Then there exists a serial model $\mathfrak{M}$ and a state $s$ such that $V_{s}^{\mathfrak{M}}(\varphi)=1.$
	\item[(ii)] Let $\Phi$ be a $\textbf{B\L}_4$-consistent set and $\Phi \vdash_{\textbf{B\L}_\textbf{4}}\varphi$. Then there exists a transitive model $\mathfrak{M}$ and a state $s$ such that $V_{s}^{\mathfrak{M}}(\varphi)=1.$
	\item[(iii)] Let $\Phi$ be a $\textbf{B\L}_T$-consistent set and $\Phi \vdash_{\textbf{B\L}_\textbf{T}}\varphi$. Then there exists a reflexive model $\mathfrak{M}$ and a state $s$ such that $V_{s}^{\mathfrak{M}}(\varphi)=1.$
	\item[(iv)] Let $\Phi$ be a $\textbf{B\L}_x$-consistent set and $\Phi \vdash_{\textbf{B\L}_\textbf{x}}\varphi$, where $x$ is a finite sequence over $\{\textbf{D, 4, T}\}$. Then there exists an appropriate model $\mathfrak{M}$ corresponding to $x$ and a state $s$ such that $V_{s}^{\mathfrak{M}}(\varphi)=1.$
	
\end{itemize}
\end{corollary}
\begin{proof}
For  part (i), we prove that if we have  (\L$_B$2) as an axiom, then the  model $\mathfrak{M}$ constructed in Theorem \ref{model_e} is serial. So, it is enough to show that for each maximal and $\textbf{B\L}_\textbf{D}$-consistent set $\Phi$, the set $\Phi\backslash B $ is consistent and therefore by Lemma \ref{maximal} it can be contained in a maximal and consistent set $\Psi$ such that $r (s_{\Phi}, s_{\Psi})=1$.
For the sake of contradiction assume  $\Phi\backslash B $ is not consistent and there is a finite $\Gamma = \{\psi_1, \cdots, \psi_n\} \subset \Phi\backslash B $  such that $\Gamma$ is not consistent, then    $\Gamma' = \{ \psi_1, \cdots,  \psi_n,  \neg\bot\}$ is not consistent, by Lemma \ref{lemma_consistent}.  Thus we have
$$\begin{array}{lll}
	(1)&\vdash_{\textbf{B\L}_D} \neg (\psi_1 \,\&\, \cdots \,\&\, \psi_n \,\&\, \neg \bot) & \text{inconsistency of } \Gamma' \\
	(2)&\vdash_{\textbf{B\L}_D} \neg \psi_1 \veebar \cdots \veebar \neg \psi_n \veebar \neg \neg \bot & (1) (\text{\L}5), (\text{R}_{\text{MP}}) \\
	(3)& \vdash_{\textbf{B\L}_D} \psi_1 \rightarrow (\neg \psi_2 \veebar \cdots \veebar \neg \psi_n \veebar \bot)& (2),(\text{\L}7), \text{replacement}, (\text{R}_{\text{MP}})\\
	(4)& \vdots & \\
	(5)&\vdash_{\textbf{B\L}_D} \psi_1 \rightarrow (\psi_2\rightarrow (\cdots (\psi_n \rightarrow \bot))) & \text{similar procedure } \\
	(6) &\vdash_{\textbf{B\L}_D} B  (\psi_1 \rightarrow (\psi_2\rightarrow (\cdots (\psi_n \rightarrow \bot)))) & (5), (R_{GB}) \\
	(7)&\vdash_{\textbf{B\L}_D} ( B  \psi_1 \,\&\, B  (\psi_1 \rightarrow (\psi_2\rightarrow (\cdots (\psi_n \rightarrow \bot))))  & \\
	& \qquad\qquad\qquad\qquad\qquad\qquad\rightarrow B ((\psi_2\rightarrow (\cdots (\psi_n \rightarrow \bot)))& (\text{\L}_B 1)
\end{array}$$
By Lemma \ref{maximal} and (6) we have:
$$B  (\psi_1 \rightarrow (\psi_2\rightarrow (\cdots (\psi_n \rightarrow \bot))))\in \Phi,$$
and similarly, by (7) we have:
$$( B  \psi_1 \,\&\, B  (\psi_1 \rightarrow (\psi_2\rightarrow (\cdots (\psi_n \rightarrow \bot)))) \rightarrow B ((\psi_2\rightarrow (\cdots (\psi_n \rightarrow \bot))) \in \Phi. $$
Also, from $\Gamma\subset \Phi\backslash B $  we have $B  \psi_1 \in \Phi$, thus $B ((\psi_2\rightarrow (\cdots (\psi_n \rightarrow \bot)) \in \Phi$. By a similar argument, we obtain that $B  \bot \in \Phi$, which is a contradiction to this assumption that (\L$_B$2) belongs to the  consistent set $\Phi$.

For part (ii), let $\Phi$ be $\textbf{B\L}_4$-consistent.
 We  show that if $s_{\Theta_1} , s_{\Theta_2}, s_{\Theta_3} \in S $ in model $\mathfrak{M}$, then we have 
$$r  (s_{\Theta_1}, s_{\Theta_3}) \geq \min \{r  (s_{\Theta_1}, s_{\Theta_2}), r  (s_{\Theta_2}, s_{\Theta_3})\}.$$
We prove that if $r  (s_{\Theta_1}, s_{\Theta_2}) = r  (s_{\Theta_2}, s_{\Theta_3}) = 1$, and so we have $r  (s_{\Theta_1}, s_{\Theta_3})=1$. The other cases are easy. From the assumption we have:
\begin{align}
	& \{\varphi \,|\, B \varphi \in \Theta_1\} \subseteq \Theta_2, \label{eq005}\\
	& \{\varphi \,|\, B \varphi \in \Theta_2\} \subseteq \Theta_3. \label{eq012}
\end{align}
From $B  \varphi \in \Theta_1$ and axiom $(\text{\L}_B3)$, we have $B  B  \varphi \in \Theta_1$. So $B \varphi \in \Theta_2$ by equation \ref{eq005}, and we obtain that $\varphi \in \Theta_3$ by equation \ref{eq012}, thus $\{\varphi \,|\, B  \varphi \in \Theta_1\}\subseteq \Theta_3$, and we have $r  (s_{\Theta_1}, s_{\Theta_3})=1$.

For part (iii), suppose that $\Phi$ is $\textbf{B\L}_T$-consistent and let $\Theta$ be a maximal and consistent set. We have $\{\varphi\,|\, B  \varphi \in \Theta \} \subseteq \Theta$, and then $r (s_{\Theta}, s_{\Theta}) = 1$.
\\
Case (iv) can be obtained by a similar discussion as the above parts.
\end{proof}

\begin{remark}
	Note that in axiomatic systems containing axiom (\L$_B$5), we only have soundness and not completeness. 
\end{remark}

\begin{theorem} \label{thm00} \textbf{(Completeness)}

\begin{itemize}
	\item[(i)] If $ \vDash\varphi$, then $\vdash_{\textbf{B\L}} \varphi$.
	\item[(ii)]  
	Let $\Gamma$ be a finite set of $D\L L_B$-formulas. Suppose that $\mathcal{M}$ is a class of serial models, and for all $(\mathfrak{M}, s) \in \mathcal{M}$, we have $(\mathfrak{M}, s) \vDash \Gamma$ and $\Gamma \vDash \varphi$. Then it follows that $\Gamma \vdash_{\textbf{B\L}_\textbf{D}} \varphi$.
	\item[(iii)] 	
	Let $\Gamma$ be a finite set of $D\L L_B$-formulas. If $\mathcal{M}$ is a class of transitive models such that for all $(\mathfrak{M}, s) \in \mathcal{M}$, it holds that $(\mathfrak{M}, s) \vDash \Gamma$ and $\Gamma \vDash \varphi$, then it follows that $\vdash_{\textbf{B\L}_\textbf{4}} \varphi$
	\item[(iv)] Let $\Gamma$ be a finite set of $D\L L_B$-formulae. If $\mathcal{M}$ is a class of reflexive models such that for all $(\mathfrak{M}, s)\in \mathcal{M}$, we have $(\mathfrak{M}, s)\vDash\Gamma$ and $\Gamma \vDash \varphi$, then $\vdash_{\textbf{B\L}_\textbf{T}}\varphi$.
	\item[(v)] Let $\Gamma$ be a finite set of $D\L L_B$-formulae. 	Let $x$ be a string containing \textbf{D}, \textbf{4}, or \textbf{T}, and let $\mathcal{M}$ have accessibility relation properties corresponding to the axioms specified in $x$ such that for all $(\mathfrak{M}, s) \in \mathcal{M}$, it holds that $(\mathfrak{M}, s) \vDash \Gamma$. If $\Gamma \vDash \varphi$, then $\vdash_{\textbf{B\L}_\textbf{x}}\varphi$.

\end{itemize}
\end{theorem}
\begin{proof}
We only prove part (i), the other parts can be concluded similarly by using Corollary \ref{corollary_model_e}.
For the sake of contradiction assume we have $\nvdash_{\textbf{B\L}} \varphi$. Thus by Theorem \ref{consistent_e}, $ \{\neg \varphi\}$ is a ${\textbf{B\L}}$-consistent set, and by Lemma \ref{maximal} we have a maximal and ${\textbf{B\L}}$-consistent set $\Phi^*$ which contains $\{\neg \varphi\}$. Therefore there exists a model $\mathfrak{M}$ and a state $s$ such that $V_s^\mathfrak{M}(\neg \varphi)=1$ by Theorem \ref{model_e}. But this is a contradiction to $\vDash \varphi$. 
\end{proof}
	
	\section{Conclusion} \label{conclusion}
	
		We suggested a doxastic \L ukasiewicz logic and some of its extensions corresponded to some traditional axioms in classical epistemic logic. Our proposed semantics is based on  Kripke models where accessibility relations and atomic propositions take infinite values in the MV-algebra [0,1]. Moreover, we proved that some proposed axiomatic systems are sound and complete corresponding to the appropriate classes of models.
	
%
	%
	%

\end{document}